\documentclass[%
 reprint,
 amsmath,amssymb,
 aps,
 pra
]{revtex4-1}

\usepackage{graphicx}
\usepackage{dcolumn}
\usepackage{bm}
\usepackage{hyperref}
\usepackage{amsthm}
\usepackage{qcircuit}
\usepackage{graphicx}
\usepackage{tikz}
\usetikzlibrary{decorations.pathreplacing}
\usepackage{bbm}

\usepackage{colortbl}
\definecolor{mygreen}{RGB}{175,233,198}
\definecolor{mypurple}{RGB}{239,229,247}
\definecolor{myyellow}{RGB}{255,238,170}
\definecolor{mygray}{RGB}{240,238,228}
\definecolor{mydarkred}{RGB}{240,0,0}
\definecolor{mylightblue}{RGB}{176,224,230}
\definecolor{mylightred}{RGB}{255,204,153}
\definecolor{mylightorange}{RGB}{255,153,153}
\definecolor{mylightgray}{RGB}{224,224,224}

\newcommand{\ket}[1]{| #1 \rangle} 
\newcommand{\bra}[1]{\langle #1 |} 
\newcommand{\braket}[2]{\langle #1 \vphantom{#2} |  #2 \vphantom{#1} \rangle} 
\newcommand{\ketbra}[2]{|#1\rangle\langle#2|}
\newcommand{\e}{\mathrm{e}}
\newcommand{\x}{\mathbf{x}}

\newtheorem{definition}{Definition}
\newtheorem{theorem}{Theorem}

\newtheorem{prop}{Proposition}

\begin{document}

\title{Quantum machine learning in feature Hilbert spaces}

\author{Maria Schuld} \email{maria@xanadu.ai}

 \affiliation{Xanadu, 372 Richmond St W, Toronto, M5V 2L7, Canada}
\author{Nathan Killoran}
 \affiliation{Xanadu, 372 Richmond St W, Toronto, M5V 2L7, Canada}

\date{\today}

\begin{abstract}
The basic idea of quantum computing is surprisingly similar to that of kernel methods in machine learning, namely to efficiently perform computations in an intractably large Hilbert space. In this paper we explore some theoretical foundations of this link and show how it opens up a new avenue for the design of quantum machine learning algorithms. We interpret the process of encoding inputs in a quantum state as a nonlinear feature map that maps data to quantum Hilbert space. A quantum computer can now analyse the input data in this feature space. Based on this link, we discuss two approaches for building a quantum model for classification. In the first approach, the quantum device estimates inner products of quantum states to compute a classically intractable kernel. This kernel can be fed into any classical kernel method such as a support vector machine. In the second approach, we can use a variational quantum circuit as a linear model that classifies data explicitly in Hilbert space. We illustrate these ideas with a feature map based on squeezing in a continuous-variable system, and visualise the working principle with $2$-dimensional mini-benchmark datasets. 
\end{abstract}


\maketitle
\section{Introduction}

The goal of many quantum algorithms is to perform efficient computations in a Hilbert space that grows rapidly with the size of a quantum system. `Efficient' means that the number of operations applied to the system grows at most polynomially with the system size. An illustration is the famous quantum Fourier transform applied to an $n$-qubit system, which uses $\mathcal{O} (\text{poly} (n))$ operations to perform a discrete Fourier transform on $2^n$ amplitudes. In continuous-variable systems this is pushed to the extreme, as a single operation -- for example, squeezing -- applied to a mode  formally manipulates a quantum state in an infinite-dimensional Hilbert space. In this sense, quantum computing can be understood as a technique to perform ``implicit'' computations in an intractably large Hilbert space through the efficient manipulation of a quantum system. \\

In machine learning, so-called \textit{kernel methods} are a well-established field with a surprisingly similar logic. In a nutshell, the idea of kernel methods is to formally embed data into a higher- (and sometimes infinite-) dimensional \textit{feature space} in which it becomes easier to analyse (see 
Figure \ref{Fig:linear}). A popular example is a support vector machine that draws a decision boundary between two classes of datapoints  by mapping the data into a feature space where it becomes linearly separable. The trick is that the algorithm never explicitly performs computations with vectors in feature space, but uses a so-called \textit{kernel function} that is defined on the domain of the original input data. Just like quantum computing, kernel methods therefore perform implicit computations in a possibly intractably large Hilbert space through the efficient manipulation of data inputs.\\

Besides this apparent link, kernel methods have been hardly studied in the quantum machine learning literature, a field that (in the definition we employ here) investigates the use of quantum computing as a resource for machine learning. Across the approaches in this young field, which vary from sampling \cite{verdon17, amin15, benedetti16b, low14, wittek17} to quantum optimisation \cite{denchev12,ogorman15}, linear algebra solvers \cite{wiebe12, rebentrost14, schuld16prediction} and using quantum circuits as trainable models for inference \cite{wan17, farhi18}, a lot of attention has been paid to recent trends in machine learning such as deep learning and neural networks. Kernel methods, which were most successful in the 1990s, are only mentioned in a few references \cite{rebentrost14, schuld17ibm}. Besides a single study on the connection between coherent states and Gaussian kernels \cite{chatterjee16}, their potential for quantum computing remains widely unexplored.  \\

The aim of this paper is to investigate the relationship between feature maps, kernel methods and quantum computing. We interpret the process of encoding inputs into a quantum state as a feature map which maps data into a potentially vastly higher-dimensional feature space, the Hilbert space of the quantum system. Data can now be analysed in this `feature Hilbert space', where simple classifiers such as linear models may gain enormous power. Furthermore, it is well known that the inner product of two data inputs that have been mapped into feature space gives rise to a kernel function that measures the distance between the data points. Kernel methods use these kernel functions to create models that have been very successful in pattern recognition. By switching between kernels one effectively switches between different models, which is known as the \textit{kernel trick}. In the quantum case, the kernel trick corresponds to changing the data encoding strategy. \\

These two perspectives, namely of kernels on the one hand and feature spaces one the other hand, naturally lead to two ways of building quantum classifiers for supervised learning. The \textit{implicit approach} takes a classical model that depends on a kernel function, but uses the quantum device to evaluate the kernel, which is computed as the inner products of quantum states in `feature Hilbert space'. The \textit{explicit approach} uses the quantum device to directly learn a linear decision boundary in feature space by optimising a variational quantum circuit.   \\

A central result of this paper is that the idea of embedding data into a quantum Hilbert space opens up a promising avenue to quantum machine learning, in which we can generically use quantum devices for pattern recognition. The implicit and explicit approaches are not only hardware-independent, but also suitable for intermediate-term quantum technologies, which allows us to test them with the generation of quantum computers that is currently being developed. Nonlinear feature maps also circumvent the need to implement nonlinear transformations on amplitude-encoded data, and thereby solve an outstanding problem in quantum machine learning which we will come back to in the conclusion.

\begin{figure}[t]
\centering
\includegraphics[width=0.4\textwidth]{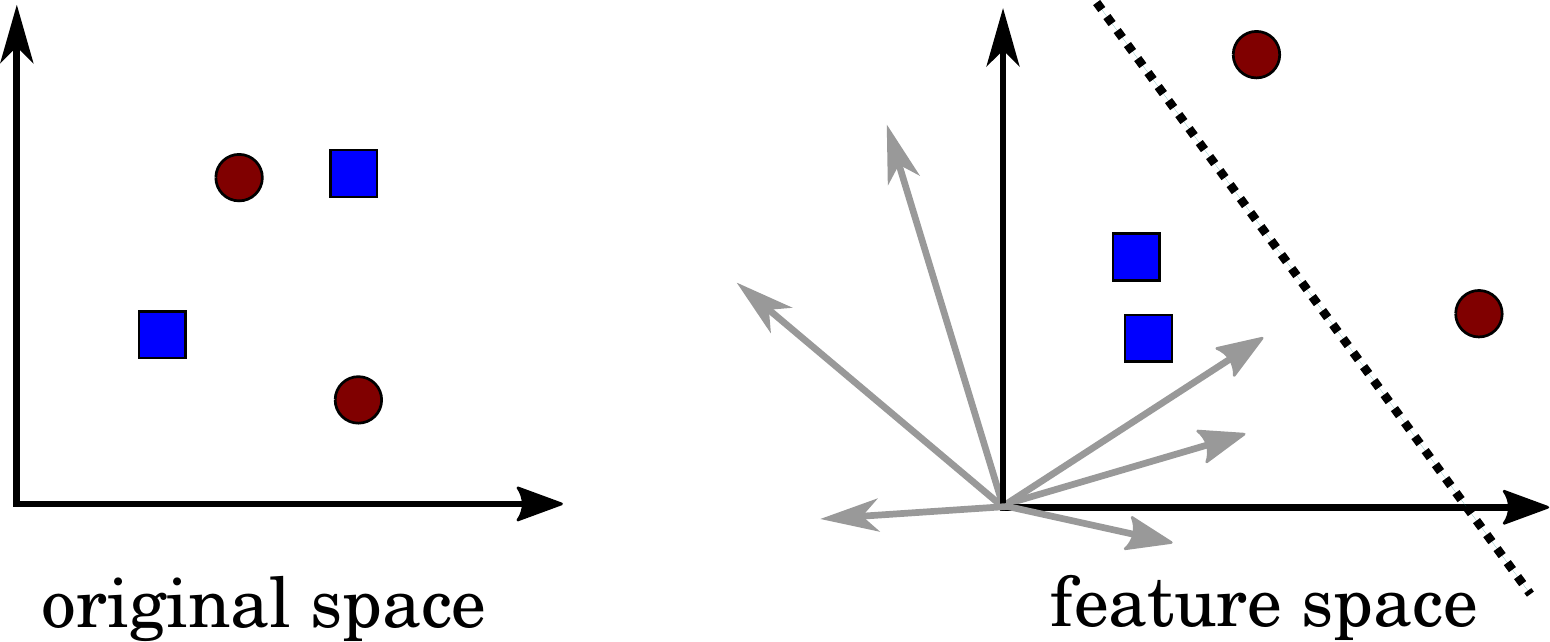}
\caption{While in the original space of training inputs, data from the two classes `blue squares' and `red circles' are not separable by a simple linear model (left), we can map them to a higher dimensional feature space where a linear model is indeed sufficient to define a separating hyperplane that acts as a decision boundary (right).}
\label{Fig:linear}
\end{figure}

\section{Feature maps, kernels and quantum computing}

In machine learning we are typically given a dataset of inputs $\mathcal{D} = \{x^1,...,x^M\}$ from a certain input set $\mathcal{X}$, and have to recognise patterns to evaluate or produce previously unseen data. Kernel methods use a distance measure $\kappa(x,x')$ between any two inputs $x,x' \in \mathcal{X}$ in order to construct models that capture the properties of a data distribution. This distance measure is connected to inner products in a certain space, the \textit{feature space}. Besides many practical applications, the most famous being the support vector machine, these methods have a rich theoretical foundation \cite{scholkopf02} from which we want to highlight some relevant points. 

\subsection{Feature maps and kernels}\label{Sec:kerneltheory}
Let us start with the definition of a feature map.
\begin{definition} \label{Def:featmap} 
Let $\mathcal{F}$ be a Hilbert space, called the {\normalfont feature space}, $\mathcal{X}$ an input set and $x$ a sample from the input set. A {\normalfont feature map} is a map $\phi:\mathcal{X} \rightarrow \mathcal{F}$ from inputs to vectors in the Hilbert space. The vectors $\phi(x) \in \mathcal{F}$ are called {\normalfont feature vectors}.
\end{definition}
\noindent Feature maps play an important role in machine learning, since they map any type of input data into a space with a well-defined metric. This space is usually of much higher dimension. If the feature map is a nonlinear function it changes the relative position between data points (as in the example of Figure \ref{Fig:linear}), and a dataset can become a lot easier to classify in feature space. Feature maps are intimitely connected to kernels \cite{berg84}.
\begin{definition} \label{Def:fkernel} 
Let $\mathcal{X}$ be a nonempty set, called the {\normalfont input set}. A function $\kappa: \mathcal{X} \times \mathcal{X} \rightarrow \mathbb{C} $ is called a {\normalfont kernel} if the Gram matrix $K$ with entries $K_{m,m'} =\kappa(x^m,x^{m'}) $ is positive semidefinite, in other words, if for any finite subset $\{x^1,...,x^M\} \subseteq \mathcal{X}$ with $M \geq 2$ and $c_1,...,c_M \in \mathbb{C}$,
\[ \sum\limits_{m,m'=1}^M c_m c_{m'}^*  \kappa(x^m, x^{m'}) \geq 0 .\]
\end{definition}
\noindent By definition of the inner product, every feature map gives rise to a kernel. 
\begin{theorem} \label{Thm:kernel} 
Let $\phi:\mathcal{X} \rightarrow \mathcal{F}$ be a feature map. The inner product of two inputs mapped to feature space defines a kernel via
\begin{equation} \kappa(x,x') := \langle \phi(x) , \phi(x') \rangle_{\mathcal{F}}, \label{Eq:featmap_kernel} \end{equation}
where $\langle\cdot,\cdot\rangle_{\mathcal{F}}$ is the inner product defined on $\mathcal{F}$.
\end{theorem}
\begin{proof}
We must show that the Gram matrix of this kernel is positive definite. For arbitrary $c_m, c_{m'} \in \mathbb{C}$ and any $\{x^1,...,x^M\} \subseteq \mathcal{X}$ with $M \geq 2$, we find that
\begin{eqnarray*}
 \sum\limits_{m,m' =1}^M c_m c_{m'}^* \kappa(x_m,x_{m'}) 
&=& \langle \sum_m c_m \phi(x_m), \sum_{m'} c_{m'} \phi(x_{m'}) \rangle \\
&=& ||\sum_m c_m \phi(x_m)||^2 \geq  0
\end{eqnarray*}
\end{proof}
The connection between feature maps and kernels means that every feature map corresponds to a distance measure in input space by means of the inner product of feature vectors. It also means that we can compute inner products of vectors mapped to much higher dimensional spaces by computing a kernel function, which may be computationally a lot easier.

\subsection{Reproducing kernel Hilbert spaces}
Kernel theory goes further and defines a unique Hilbert space associated with each kernel, the \textit{reproducing kernel Hilbert space} or RKHS \cite{hofmann08, aronszajn50}. Although rather abstract, this concept is useful in order to understand the significance of kernels for machine learning, as well as their connection to linear models in feature space.
\begin{definition}\label{Def:rkhs} 
Let $\mathcal{X}$ be a non-empty input set and $\mathcal{R}$ a Hilbert space of functions $f: \mathcal{X} \rightarrow \mathbb{C}$ that map inputs to the real numbers. Let $\langle \cdot , \cdot \rangle $ be an inner product defined on $\mathcal{R}$ (which gives rise to a norm via $||f||= \sqrt{\langle f, f \rangle}$). $\mathcal{R}$ is a {\normalfont reproducing kernel Hilbert space} if every point evaluation is a continuous functional $F: f \rightarrow f(x)$ for all $x \in \mathcal{X}$. This is equivalent to the condition that there exists a function $\kappa: \mathcal{X} \times \mathcal{X} \rightarrow \mathbb{C}$ for which
\begin{equation}
\langle f, \kappa(x, \cdot) \rangle = f(x) 
\label{Eq:repro}
\end{equation} 
with $\kappa(x, \cdot) \in \mathcal{R}$ and for all  $f \in \mathcal{H}$, $x\in \mathcal{X}$. 
\end{definition}
\noindent The function $\kappa$ is the unique \textit{reproducing kernel} of $\mathcal{R}$, and Eq. (\ref{Eq:repro}) is the \textit{reproducing property}. 
Note that a different, but isometrically isomorphic Hilbert space can be derived for a so-called Mercer kernel \cite{mercer09}.\\

\begin{figure}[t]
\begin{tikzpicture}
\path (-3,0) node[fill = white!95!black, rectangle, rounded corners,   draw=black, thick] (f) {feature map};
\path (0,0) node[fill = white!95!black, rectangle, rounded corners,   draw=black, thick] (k) {kernel};
\path (3,0) node[fill = white!95!black, rectangle, rounded corners,   draw=black, thick] (r) {RKHS};
\draw[->, thick] (f) -- node [above] {Thm 1} (k) ;
\draw [<->, thick] (k) -- node [above]{Def 3} (r) ;
\draw[->, thick](f.south east) to [in =220, out=320] node[above]{Thm 2}  (r.south west) ;

\end{tikzpicture}
\caption{Relationships between the concepts of a feature map, kernel and reproducing kernel Hilbert space. 
}
\label{Fig:flow}
\end{figure}
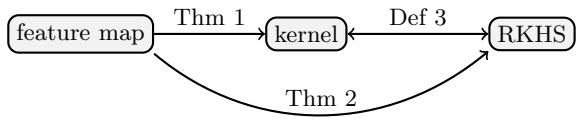

Since a feature map gives rise to a kernel and a kernel gives rise to a reproducing kernel Hilbert space, we can construct a unique reproducing kernel Hilbert space for any given feature map (see Figure \ref{Fig:flow}).
\begin{theorem} \label{Thm:featmap_rkhs} 
Let $\phi:\mathcal{X} \rightarrow \mathcal{F}$ be a feature map over an input set $\mathcal{X}$, giving rise to a complex kernel $\kappa(x,x') = \langle \phi(x) , \phi(x') \rangle_{\mathcal{F}}$. The corresponding reproducing kernel Hilbert space has the form
\begin{multline}\mathcal{R}_{\kappa} = \{f:\mathcal{X}\rightarrow \mathbb{C} | \; \\ f(x) = \langle w, \phi(x) \rangle_{\mathcal{F}}, \; \forall x\in \mathcal{X},w \in \mathcal{F} \} \label{Eq:featmap_rkhs} \end{multline}
\end{theorem}
The functions $\langle w, \cdot \rangle$ in the RKHS associated with feature map $\phi$ can be interpreted as linear models, for which $w \in \mathcal{F}$ defines a hyperplane in feature space.\\

In machine learning these rather formal concepts gain relevance because of the (no less formal) \textit{representer theorem} \cite{scholkopf01}: 
\begin{theorem}\label{Thm:represth}
Let $\mathcal{X}$ be an input set, $\kappa : \mathcal{X} \times \mathcal{X} \rightarrow \mathbb{R}$ a kernel, $\mathcal{D}$ a data set consisting of data pairs $(x^m,y^m) \in \mathcal{X} \times \mathbb{R}$ and $f: \mathcal{X} \rightarrow \mathbb{R}$ a class of model functions that live in the reproducing kernel Hilbert space $\mathcal{R}_{\kappa}$ of $\kappa$. Furthermore, assume we have a cost function $\mathcal{C}$ that quantifies the quality of a model by comparing predicted outputs $f(x^m)$ with targets $y^m$, and which has a regularisation term of the form $g(||f||)$ where  $g: [0,\infty ) \rightarrow \mathbb{R}$ is a strictly monotonically increasing function. Then any function $f^* \in \mathcal{R}_{\kappa}$ that minimises the cost function $C$ can be written as
\begin{equation}
 f^*(x) = \sum_{m=1}^M \alpha_m  \kappa(x, x^m),
\label{Eq:represth}
\end{equation}
for some parameters $\alpha_m \in \mathbb{R}$.
\end{theorem}
The representer theorem implies that for a common family of machine learning optimisation problems over functions in an RKHS $\mathcal{R}$, the solution can be represented as an expansion of kernel functions as in Eq. (\ref{Eq:represth}). Consequently, instead of explicitly optimising over an infinite-dimensional RKHS we can directly start with the implicit ansatz of Eq. (\ref{Eq:represth}) and solve the convex optimisation problem of finding the parameters $\alpha_m$. The combination of Theorem \ref{Thm:featmap_rkhs} and Theorem \ref{Thm:represth} shows another facet of the link of kernels and feature maps. A model that defines a hyperplane in feature space can often be written as a model that depends on kernel evaluations. In Section \ref{Sec:qml} we will translate these two viewpoints into two ways of designing quantum machine learning algorithms.

\subsection{Input encoding as a feature map}\label{Sec:encoding}

The immediate approach to combine quantum mechanics and the theory of kernels is to associate the Hilbert space of a quantum system with a reproducing kernel Hilbert space and find the reproducing kernel of the system. We show in Appendix \ref{App:quantumrkhs} that for Hilbert spaces with discrete bases, as well as for the special `continuous-basis' case of the Hilbert space of coherent states, the reproducing kernel is given by inner products of basis vectors. This insight can lead to interesting results. For example, Chatterjee et al. \cite{chatterjee16} show that the inner product of an optical coherent state can be turned into a \textit{Gaussian kernel} (also called \textit{radial basis function kernel}) which is widely used in machine learning. However, to widen the framework we choose another route here. Instead of asking what kernel is associated with a quantum Hilbert space, we associate a quantum Hilbert space with a feature space and derive a kernel that is given by the inner product of quantum states. As seen in the previous section, this will automatically give rise to an RKHS, and the entire apparatus of kernel theory can be applied.\\

Assume we want to encode some input $x$ from an input set $\mathcal{X}$ into a quantum state that is described by a vector $\ket{\phi(x)}$ and which lives in Hilbert space $\mathcal{F}$. This procedure of `input encoding' fulfills the definition of a feature map $\phi: \mathcal{X} \rightarrow \mathcal{F}$, which we call a \textit{quantum feature map} here. According to Theorem \ref{Thm:kernel} we can derive a kernel $\kappa$ from this feature map via Eq. (\ref{Eq:featmap_kernel}). By virtue of Theorem \ref{Thm:featmap_rkhs}, the kernel is the reproducing kernel of an RKHS $\mathcal{R}_{\kappa}$ as defined in Eq. (\ref{Eq:featmap_rkhs}). The functions in $\mathcal{R}_{\kappa}$ are the inner products of the `feature-mapped' input data and a vector $\ket{w} \in \mathcal{F}$, which defines a linear model 
\begin{equation} f(x;w) = \braket{w}{\phi(x)} \label{Eq:hilbertmodel} \end{equation} Note that we use Dirac brackets $\braket{\cdot}{\cdot}$ instead of the inner product $\langle \cdot,\cdot \rangle$ to signify that we are calculating inner products in a quantum Hilbert space. Finally, the representer theorem \ref{Thm:represth} guarantees that the minimiser $ \min_{w} C(w,\mathcal{D})$ of the empirical risk 
\[C(w, \mathcal{D})  = \sum_{m=1}^M |f(x^m;w)- y^m|^2 + ||f||_{\mathcal{R}_{\kappa}} \] 
can be expressed by Equation (\ref{Eq:represth}). The simple idea of interpreting $x \rightarrow \ket{\phi(x)}$ as a feature map therefore allows us to make use of the rich theory of kernel methods and gives rise to machine learning models whose trained candidates can be expressed by inner products of quantum states. Note that if the state $\ket{\phi(x)}$ has complex amplitudes, we can always construct a real kernel by taking the absolute square of the inner product.

\section{Quantum machine learning in feature Hilbert space}\label{Sec:qml}

Now let us enter the realm of quantum computing and quantum machine learning. We show how to use the ideas of Section \ref{Sec:encoding} to design two types of quantum machine learning algorithms and illustrate both approaches with an example from continuous-variable systems.

\subsection{Feature-encoding circuits}

From the perspective of quantum computing, a quantum feature map $x \rightarrow \ket{\phi(x)}$ corresponds to a state preparation circuit $U_{\phi}(x)$ that acts on a ground or vacuum state $\ket{0...0}$ of a Hilbert space $\mathcal{F}$ as $U_{\phi}(x)\ket{0...0} = \ket{\phi(x)}$. We will call $U_{\phi}(x)$ the \textit{feature-embedding circuit}. The models from Eq. (\ref{Eq:hilbertmodel}) in the reproducing Hilbert space from Definition \ref{Def:fkernel} are inner products between $\ket{\phi(x)}$ and a general quantum state $\ket{w} \in \mathcal{F}$. We therefore consider a second circuit $W$ with $W\ket{0...0} = \ket{w}$, which we call the \textit{model circuit}. The model circuit specifies the hyperplane of a linear model in feature Hilbert space. If the feature state $\ket{\phi(x)}$ is orthogonal to $\ket{w}$, then $x$ lies on the decision boundary, whereas states with a positive [negative] inner product lie on the left [right] side of the hyperplane.  \\

To show some examples of feature-embedding circuits and their associated kernels, let us have a look at popular input encoding techniques in quantum machine learning. \\

\paragraph{Basis encoding.} Many quantum machine learning algorithms assume that the inputs $x$ to the computation are encoded as binary strings represented by a computational basis state of the qubits \cite{wang15, farhi18}. For example, $x=01001$ is represented by the $5$-qubit basis state $\ket{01001}$. The computational basis state corresponds to a standard basis vector $\ket{i}$ (with $i$ being the integer representation of the bitstring) in a $2^n$-dimensional Hilbert space $\mathcal{F}$, and the effect of the feature-embedding circuit is given by
\[U_{\phi}: x \in \{0,1\}^n \rightarrow \ket{i} . \]
This feature map maps each data input to a state from an orthonormal basis and is equivalent to the generic finite-dimensional case discussed in Appendix \ref{App:quantumrkhs}. As shown there, the generic kernel is the Kronecker delta \[\kappa(x,x') = \braket{i}{j} = \delta_{ij},\]
which is a binary similarity measure that is only nonzero for two identical inputs.\\

\paragraph{Amplitude encoding.} Another approach to information encoding is to associate normalised input vectors $\x = (x_0,...,x_{N-1})^T  \in \mathbb{R}^N$ of dimension $N = 2^n$ with the amplitudes of a $n$ qubit state $\ket{\psi_{\x}}$ \cite{wiebe12, schuld17ibm},
\[U_{\phi}: \x \in \mathbb{R}^N \rightarrow  \ket{\psi_{\x}} = \sum\limits_{i=0}^{N-1} x_i \ket{i} . \]
As above, $\ket{i}$ denotes the $i$'th computational basis state. This choice corresponds to the linear kernel,
\[ \kappa(\x,\x') = \braket{\psi_{\x}}{ \psi_{\x'} } = \x^T\x'. \]
 
\paragraph{Copies of quantum states.}  With a slight variation of amplitude encoding we can implement polynomial kernels \cite{rebentrost14}. Taking $d$ copies of an amplitude encoded quantum state, 
\[ U_{\phi}: \x \in \mathbb{R}^N \rightarrow \ket{\psi_{\x}} \otimes  \cdots  \otimes \ket{\psi_{\x}} , \]
corresponds to the kernel
\[ \kappa(\x,\x') = \braket{\psi_{\x}}{\psi_{\x'} }  \cdots  \braket{\psi_{\x}}{\psi_{\x'} } = (\x^T\x')^d. \]

\paragraph{Product encoding.} One can also use a (tensor) product encoding, in which each feature of the input $\x = (x_1,..,x_N)^T \in \mathbb{R}^N$ is encoded in the amplitudes of one separate qubit. An example is to encode $x_i$ as $\ket{\phi(x_i)} = \cos(x_i) \ket{0} + \sin(x_i) \ket{1}$ for $i=1,...,N$\cite{stoudenmire16, guerreschi17}. This corresponds to a feature-embedding circuit with the effect
\[U_{\phi}: \x \in \mathbb{R}^N \rightarrow \begin{pmatrix} \cos x_1 \\  \sin x_1 \end{pmatrix} \otimes \cdots \otimes \begin{pmatrix} \cos x_N \\  \sin x_N \end{pmatrix} \in \mathbb{R}^{2^N}, \]
and implies a cosine kernel,
\[\kappa(\x,\x')= \prod_{i=1}^N \cos (x_i - x'_i).\]

\begin{figure}[t]
\centering
\includegraphics[width=0.3\textwidth]{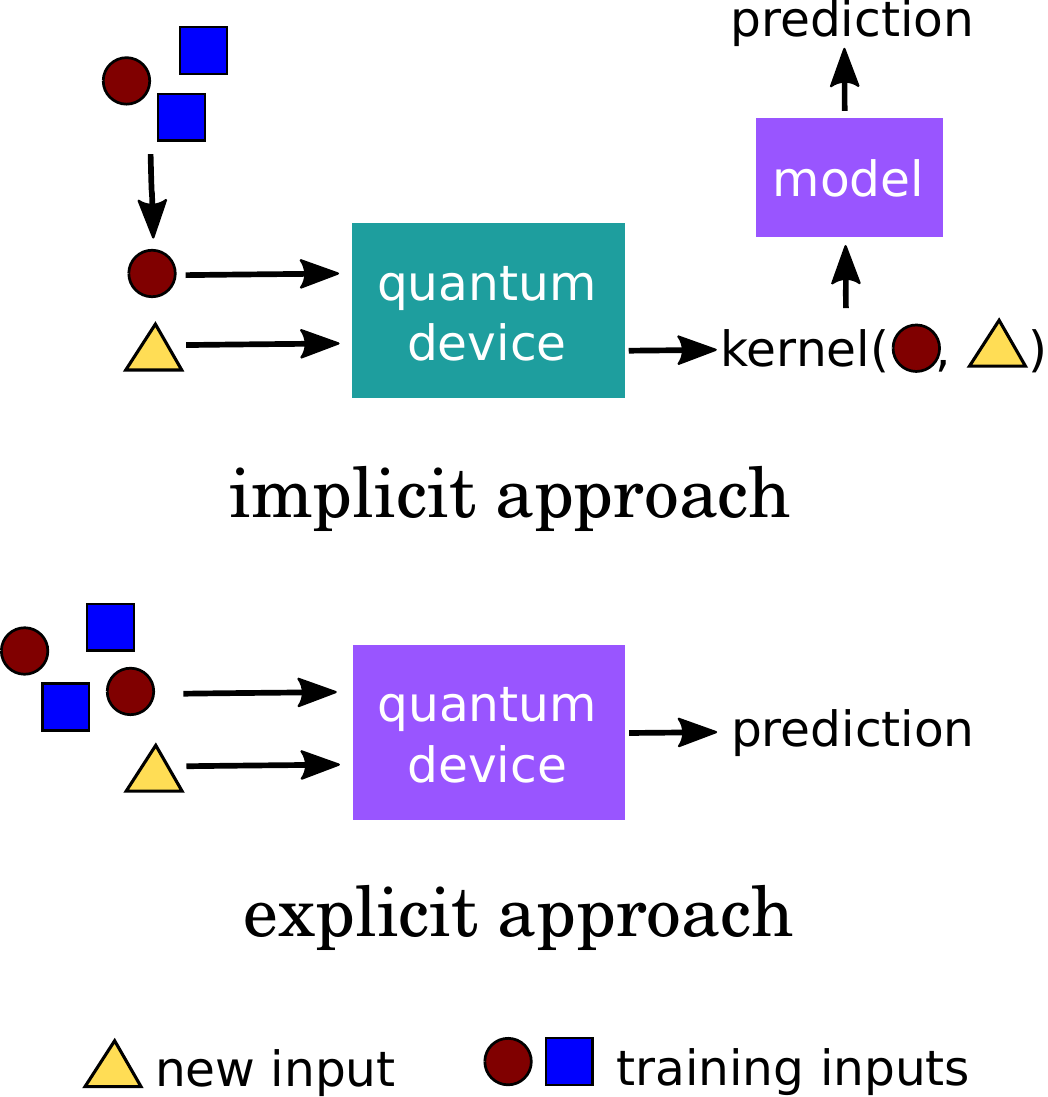}
\caption{Illustration of the two approaches to use quantum feature maps for supervised learning. The implicit approach uses the quantum device to evaluate the kernel function as part of a hybrid or quantum-assisted model which can be trained by classical methods. In the explicit approach, the model is solely computed by the quantum device, which consists of a variational circuit trained by hybrid quantum-classical methods. }
\label{Fig:approaches}
\end{figure}

\subsection{Building a quantum classifier}

Having formulated the ideas from Section \ref{Sec:encoding} in the language of quantum computing, we can identify two different strategies of designing a quantum machine learning algorithm (see Figure \ref{Fig:approaches}). On the one hand, we can use the quantum computer to estimate the inner products $\kappa(x,x') = \braket{\phi(x)}{\phi(x')}$ from a kernel-dependent model as in Eq. (\ref{Eq:represth}), which we call the \textit{implicit approach}, since we use the quantum system to estimate distance measures on input space. This strategy requires a quantum computer that can do two things: to implement $U_{\phi}(x)$ for any $x \in \mathcal{X}$ and to estimate inner products between quantum states (for example using a SWAP test routine). The computation of the model from those kernel estimates, as well as the training algorithm is left to a classical device. This is an excellent strategy in the context of intermediate-term quantum technologies \cite{preskill18}, where we are interested in using a quantum computer only for small routines of limited gate count, and compute as much as possible on the classical hardware. Note that in the long term, quantum computers could also be used to learn the parameters $\alpha_m$ by computing the inverse of the kernel Gram matrix, which has been investigated in Refs. \cite{rebentrost14, schuld16lr}. \\

On the other hand, and as motivated in the introduction, one can bypass the representer theorem and explicitly perform the classification in the `feature Hilbert space' of the quantum system. We call this the \textit{explicit approach}. For example, this can mean to find a $\ket{w}$ that defines a model \ref{Eq:hilbertmodel}. To do so, we can make the model circuit trainable, $W = W(\theta)$, so that quantum-classical hybrid training \cite{mcclean16, guerreschi17} of $\theta$ can learn the optimal model $\ket{w(\theta)} = W(\theta)\ket{0}$. The ansatz for the model circuit's architecture defines the space of possible models and can act as regularisation (see also \cite{stoudenmire16}). Below we will follow a slightly more general strategy and compute a state $W(\theta) U_{\phi} \ket{0...0}$, from which measurements determine the output of the model. Depending on the measurement, this is not necessarily a linear model in feature Hilbert space. We could even go further and include postselection in the model circuit, which might give the classifier in feature Hilbert space even more power. \\

Using quantum computers for learning tasks with these two approaches is desirable in various settings. For example, the implicit approach may be interesting in cases where the quantum device evaluates kernels or models faster in terms of absolute runtime speed. Another interesting example is a setting in which the kernel one wants to use is classically intractable because the runtime grows exponentially or even faster with the input dimension. The explicit approach may be useful when we want to leave the limits of the RKHS framework and construct classifiers directly on Hilbert space. \\

In the remainder of this work we want to explore these two approaches with several examples. We use squeezing in continuous-variable quantum systems as a feature map, for which the Hilbert space $\mathcal{F}$ is an infinite-dimensional Fock space. This constructs a squeezing-based quantum machine learning classifier which can for example be implemented by optical quantum computers. 

\subsection{Squeezing as a feature map} \label{linsep}

\begin{figure}[t]
\centering
\includegraphics[width=0.45\textwidth]{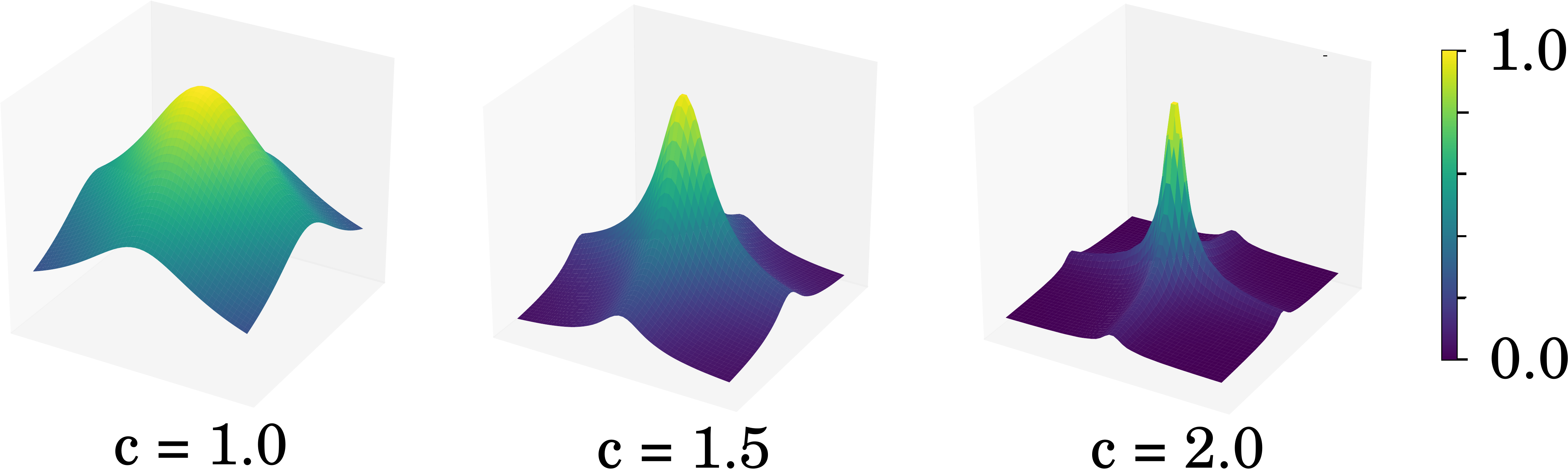}
\caption{Shape of the squeezing kernel function $\kappa_{\mathrm{sq}}(x,x')$ from Equation (\ref{Eq:squeezingkernel}) for different squeezing strength hyperparameters $c$. The input $x$ is fixed at $(0,0)$ and $x'$ is varied. The plots show the interval $[-1,1]$ on both horizontal axes.}
\label{Fig:squeezingkernel}
\end{figure}

A \textit{squeezed vacuum state} is defined as
\[ \ket{z}  = \frac{1}{\sqrt{\cosh(r)}} \sum\limits_{n=0}^{\infty} \frac{\sqrt{(2n)!}}{ 2^n n!} (-\e^{i\varphi} \tanh (r))^n \ket{2n}, \] 
where $\{ \ket{n}\}$ denotes the Fock basis and $z = r \e^{i\varphi} $ is the complex \textit{squeezing factor} with absolute value $r$ and phase $\varphi$. It will be useful to introduce the notation $\ket{z} = \ket{(r,\varphi)}$. 
We can interpret $x \rightarrow \ket{\phi(x)} = \ket{(c, x)}$ as a feature map from a one-dimensional real input space $x \in \mathbb{R}$ to the Hilbert space of Fock states, in short, the \textit{Fock space}. Here, $c$ is a constant hyperparameter that determines the strength of the squeezing, and $x$ is associated with the phase. Moreover, when given multi-dimensional inputs in a dataset of vectors $\x=( x_1, ...,x_N)^T \in \mathbb{R}^N$, we can define the joint state of $N$ squeezed vacuum modes,
\begin{equation} \phi: x \rightarrow \ket{(c, \x)},\label{Eq:squeezingmap} \end{equation}
with
\[\ket{(c, \x)} = \ket{(c, x_1)} \otimes \hdots \otimes \ket{(c, x_N)} \in \mathcal{F}, \]
as a feature map, where $\mathcal{F}$ is now a multimode Fock space. We call this feature map the \textit{squeezing feature map with phase encoding}. \\

\begin{figure}[t]
\centering
\includegraphics[width=0.40\textwidth]{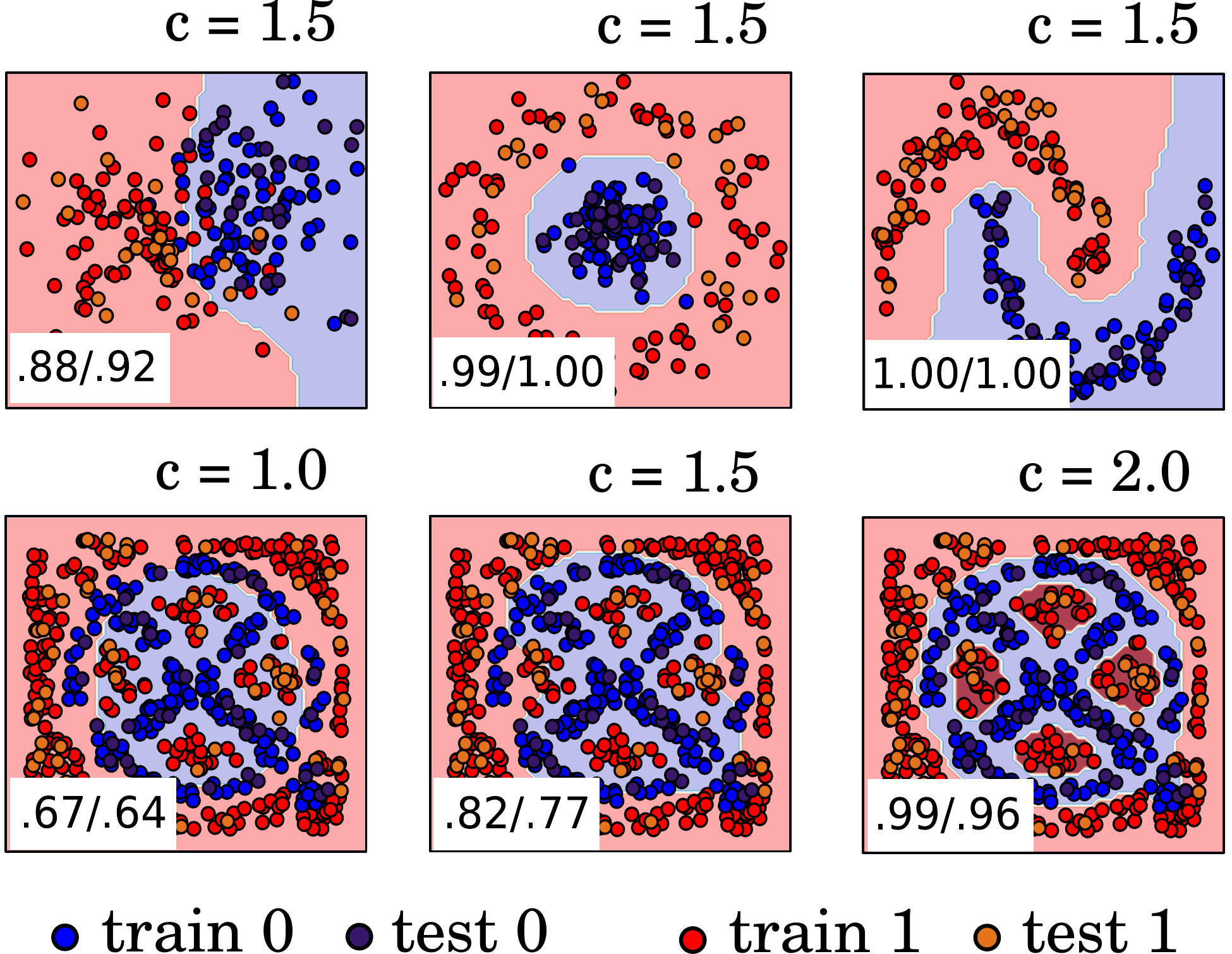}
\caption{Decision boundary of a support vector machine with the custom kernel from Eq. (\ref{Eq:squeezingkernel}). The shaded areas show the decision regions for Class 0 (blue) and Class 1 (red), and each plot shows the rate of correct classifications on the training set/test set. The first row plots three standard $2$-dimensional datasets: `circles', `moons' and `blobs', each with $150$ test and $50$ training samples. The second row illustrates that increasing the squeezing   hyperparameter $c$ changes the classification performance. Here we use a dataset of $500$ training and $100$ test samples. Training was performed with python's \textit{scikit-learn} SVC classifier using a custom kernel which implements the overlap of Eq. (\ref{Eq:squeezingkernel2}).}
\label{Fig:sq_svm}
\end{figure} 

The kernel 
\begin{equation} \kappa(\x,\x'; c) = \prod \limits^N_{i=1} \braket{(c,x_i)}{(c,x_i')}\label{Eq:squeezingkernel} \end{equation}
with
\begin{equation}\braket{(c,x_i)}{(c,x_i')} = \sqrt{\frac{\mathrm{sech}\ c \; \mathrm{sech}\ c}{1-\e^{i(x_i'- x_i ) }\tanh \ c \;\tanh \ c}},\label{Eq:squeezingkernel2} \end{equation}
derived from this feature map \cite{barnett02} is easy to compute on a classical computer. 
It is plotted in Figure \ref{Fig:squeezingkernel}, where we see that the hyperparameter $c$ determines the variance of the kernel function. Note that we can also encode features in the absolute value of the squeezing and define a \textit{squeezing feature map with absolute value encoding}, $\x \rightarrow \ket{\phi(\x)} = \ket{(\x, c)}$. However, in this version we cannot vary the variance of the kernel function, which is why we use the phase encoding in the following investiagtions.

\subsection{An implicit quantum-assisted classifier} \label{linclass}

In the implicit approach, we evaluate the kernel in Eq. (\ref{Eq:squeezingkernel}) with a quantum computer and feed it into a classical kernel method. Instead of using a real quantum device, we exploit the fact that, in the case of squeezing, the kernel can be efficiently computed classically, and use it as a custom kernel in a support vector machine. Figure \ref{Fig:sq_svm} shows that such a model easily learns the decision boundary of $2$-dimensional mini-benchmark datasets. \\

\begin{figure}[t]
\centering
\includegraphics[width=0.40\textwidth]{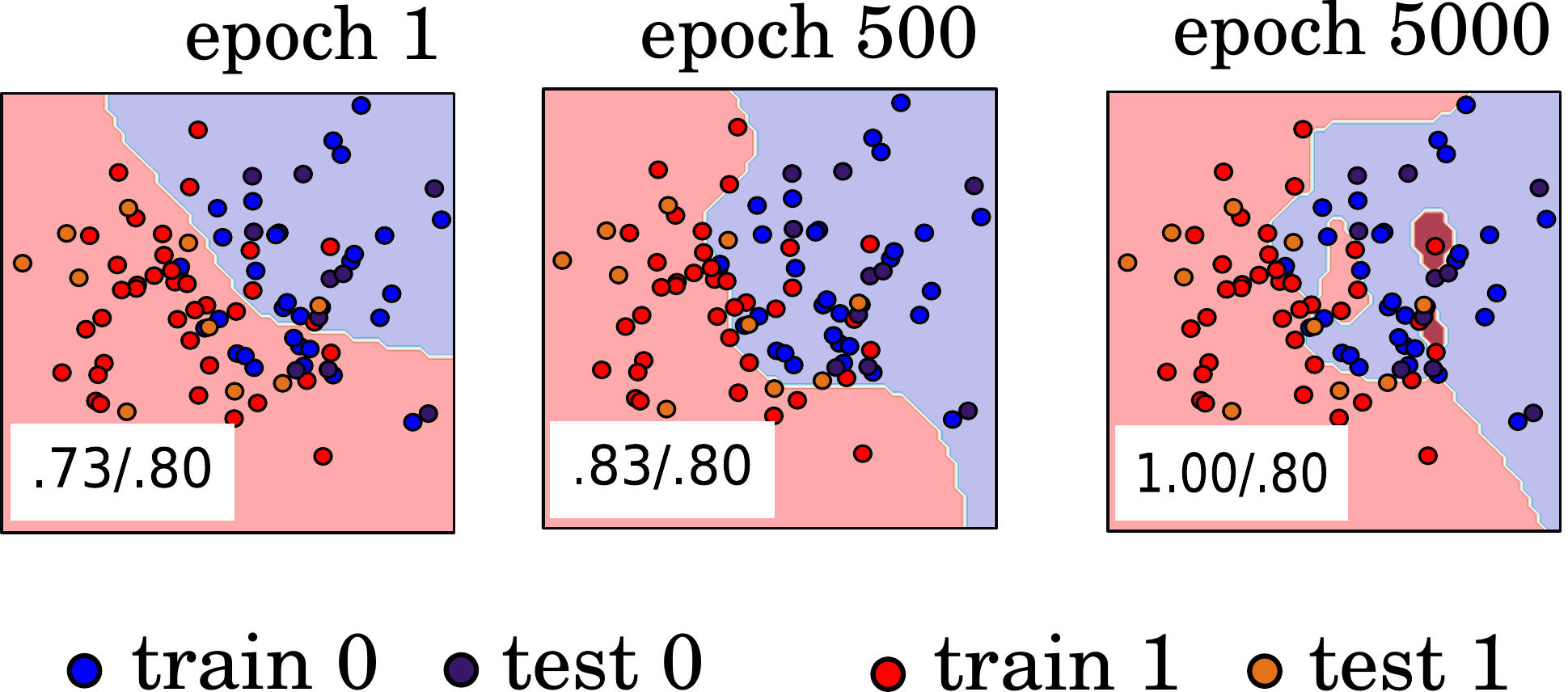}
\caption{Decision boundary of a perceptron classifier in Fock space after mapping the $2$-dimensional data points via the squeezing feature map with phase encoding from Eq. (\ref{Eq:squeezingmap}) (with $c=1.5$). The perceptron only acts on the real subspace and without regularisation. The `blobs' dataset has now only $70$ training and $20$ test samples. The perceptron achieves a training accuracy of $1$ after less than $5000$ epochs, which means that the data is linearly separable in Fock space. Interestingly, in this example the test performance remains exactly the same. The simulations were performed with the Strawberry Fields simulator as well as a scikit-learn out-of-the-box perceptron classifier.}
\label{Fig:sq_perc}
\end{figure} 

Since the idea of a support vector machine is to find the maximum-margin hyperplane in feature space, we want to know whether we can always find a hyperplane for which the training accuracy is $1$. In other words, we ask if the data becomes linearly separable in Fock space by the squeezing feature map. An easy way to do this is to apply a perceptron classifier to the data in feature space. The perceptron is guaranteed to find such a separating hyperplane if it exists. Figure \ref{Fig:sq_perc} shows the performance of a perceptron classifier in the Fock space for the `blobs' data. The data was mapped to this space by the squeezing feature map with phase encoding. As one can see, after $5000$ epochs (runs through the dataset) the decision boundary perfectly fits the training data, achieving an accuracy of $1$. The number of iterations to train the perceptron is known to increase with $\mathcal{O}(1/\gamma^2)$ where $\gamma$ is the margin between the two classes \cite{novikoff63}, and indeed we find in other simulations that the `moons' and `circles' data only take a few epochs until reaching full accuracy. Although the perfect fit to the training data is of course not useful for machine learning (as can be seen by the non-increasing accuracy on the test set) these results are a clue to the fact that the squeezing feature map makes data linearly separable in feature space, a fact that we prove in Appendix \ref{App:linsep}.\\

While the results of the simulations are promising, a goal is to find more sophisticated kernels. Although quantum computers could offer constant speed advantages, they become indispensable if the feature map circuit is classically intractable. However, squeezed states are an example of so-called Gaussian states, and it is well known that Gaussian states (although living in an infinite-dimensional Hilbert space) can be efficiently simulated by a classical computer \cite{bartlett02a}, which we used in the simulations. In order to do something more interesting, one needs non-Gaussian elements to the circuit. For example, one can extend a standard linear optical network of beamsplitters by a cubic phase gate \cite{gottesman01, lloyd99} or use photon number measurements \cite{bartlett02b}. To this end, let $V_{\phi}(\x)$ be a non-Gaussian feature map circuit, i.e. a quantum algorithm that takes a vacuum state and prepares an $\x$-dependent non-Gaussian state. The kernel
\[ \kappa(\x, \x') = \bra{0...0}V^{\dagger}_{\phi}(\x)V_{\phi}(\x')\ket{0...0}  \]
can in general not be simulated by a classical computer any more. It is therefore an interesting open question what type of feature map circuits $V_{\phi}$ are classically intractable, but at the same time lead to powerful kernels for classical models such as support vector machines.

\begin{figure}[t]
\centering
\begin{flushleft} 
a.)
\end{flushleft}
\begin{tikzpicture}[every node/.style={inner sep=0.8,outer sep=0.}]

\draw [decorate,decoration={brace,amplitude=5pt}]
(-2,0.5) -- (-0.1,0.5) node [black,align=center,midway,yshift=0.8cm] 
{feature map \\ circuit};
\draw [decorate,decoration={brace,amplitude=5pt},xshift=0pt,yshift=0pt]
(0.1,0.5) -- (2,0.5) node [black,midway,align=center,yshift=0.8cm] 
{model \\ circuit $W(\theta)$};

\path (-2,-0.) node[draw, shape=circle, anchor = west] (zz1) {$x_1$};
\path (-2,-0.8) node[draw, shape=circle, anchor = west](zz2) {$x_2$};

\path (0,0) node[draw, shape=circle, anchor = center] (ii0) {\phantom{$h$}};
\path (0,-0.5) node[draw, shape=circle, anchor = center](ii1) {\phantom{$h$}};
\path (0,-1) node[draw, shape=circle, anchor = center] (ii2) {\phantom{$h$}};
\path (0,-1.5) node[draw, shape=circle, anchor = center] (ii3) {\phantom{$h$}};
\path (0,-2) node[shape=circle,anchor = center] (ii4) {$\vdots$};

\path (2,-0.) node[draw, shape=circle, anchor = west] (oo1) {$o_0$};
\path (2,-0.8) node[draw, shape=circle, anchor = west](oo2) {$o_1$};

\path (3,-0.) node[fill = white, anchor = west] (pp1) {$p(y=0)$};
\path (3,-0.8) node[fill = white, anchor = west](pp2) {$p(y=1)$};

\foreach \i in {0,1,2,3,4}{
	\foreach \j in {1,2}{
		\draw (zz\j)--(ii\i);
		}}

\foreach \i in {0,1,2,3,4}{
	\foreach \j in {1,2}{
		\draw (ii\i)--(oo\j);
		}}
		
\draw (oo1)--(2.9,-0.);
\draw (oo2)--(2.9,-0.8);
\draw (oo2)--(2.9,-0.);
\draw (oo1)--(2.9,-0.8);
\path (-2,-1.4) node[ align = center] (is) {$\mathcal{X}$};
\path (0.3,-2.4) node[align = center] (os) {$\mathcal{F}$};
\path (2.5,-1.4) node[align = center] (fs) {$\mathcal{Y}$};

\end{tikzpicture}
\begin{flushleft} 
b.)
\end{flushleft}
$$\qquad
\Qcircuit @C=1em @R=.7em {
\lstick{\ket{(c,x_1)}}   & \multigate{1}{W(\theta)} & \qw  &\measureD{p(n_1)} \\
\lstick{\ket{(c,x_2)}} &  \ghost{W(\theta)} & \qw &\measureD{p(n_2)} \\
}
$$

\begin{flushleft}
c.)
\end{flushleft}
$$
\Qcircuit @C=0.6em @R=.7em {
& \multigate{1}{BS(\theta_1, \theta_2 )} & \gate{D(\theta_3)} & \gate{P(\theta_5)} & \gate{V(\theta_7)} & \qw \\
&  \ghost{BS(\theta_1, \theta_2 )} & \gate{D(\theta_4)} & \gate{P(\theta_6)} & \gate{V(\theta_8)} & \qw\\
}
$$

\caption{a.) Representation of the Fock-space-classifier in the graphical language of quantum neural networks. A vector $(x_1, x_2)^T$ from the input space $\mathcal{X}$ gets mapped into the feature space $\mathcal{F}$ which is the infinite-dimensional $2$-mode Fock space of the quantum system. The model circuit, including photon detection measurement, implements a linear model in feature space and reduces the ``infinite hidden layer'' to two outputs. b.) The model circuit of the explicit classifier described in the text uses only $2$ modes to instantiate this infinite-dimensional hidden layer. The variational circuit $W(\theta)$ consists of repetitions of a gate block. We use the gate block shown in c.) with the beamsplitter (BS), displacement (D), quadratic (P) and cubic phase gates (C) described in the text.  }
\label{Fig:graphicandcircuit}
\end{figure}

\subsection{An explicit quantum classifier}

In the explicit approach defined above, we use a parametrised continuous-variable circuit $W(\theta)$ to build a ``Fock-space'' classifier. For our squeezing example this can be done as follows. We start with two vacuum modes $\ket{0}\otimes \ket{0}$. To classify a data input $\x$, first map the input to a quantum state $\ket{c, \x}  = \ket{c, x_1} \otimes \ket{c, x_2}$ by performing a squeezing operation on each of the  modes. Second, apply the model circuit $W(\theta)$ to $\ket{c, \x}$. Third, interpret the probability $p(n_1, n_2)$ of measuring a certain Fock state $\ket{n_1, n_2}$ as the output of the machine learning model. Since this probability depends on the displacement and squeezing intensity, it is better to define two probabilities, say $p(n_1=2,n_2=0)$ and $p(n_1=0,n_2=2)$, as a one-hot encoded output vector $(o_0, o_1)$. This output vector can be normalised \footnote{In contrast to a standard technique in machine learning, it is not advisable to use a softmax layer for this purpose, since $o_0, o_1$ can be very small, which leads to almost uniform probabilities.} to a new vector
\[   \frac{1}{o_0 + o_1} \begin{pmatrix} o_0\\ o_1\end{pmatrix} = \begin{pmatrix} p(y=0)\\ p(y=1)\end{pmatrix},\]
where $p(y=0)$, $p(y=1)$ can now be interpreted as the probability for the model to predict class $y=0$ and $y=1$, respectively. The final label is the class with the higher probability. We can interpret this circuit in the graphical representation of neural networks as shown at the top in Figure \ref{Fig:graphicandcircuit}.   \\

\begin{figure}[t]
\centering
\includegraphics[width = 0.45\textwidth]{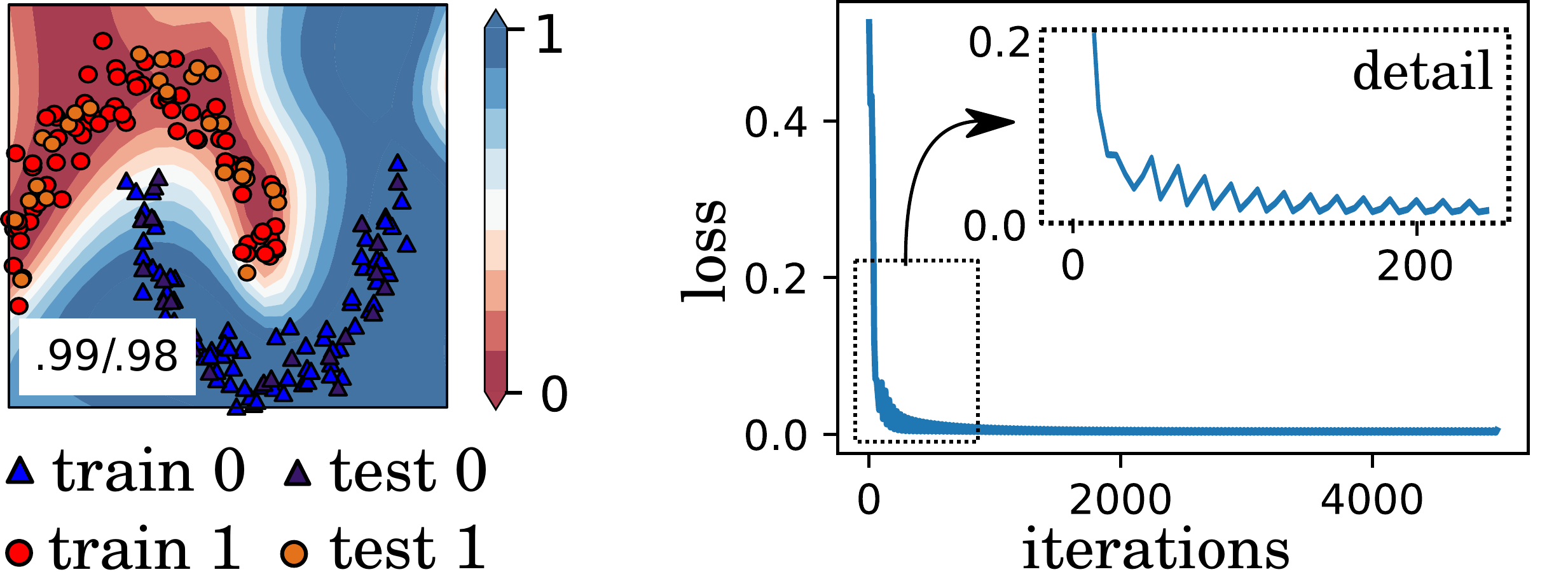}
\caption{Fock space classifier presented in Figure \ref{Fig:graphicandcircuit} and the text for the `moons' dataset. The shaded areas show the probability $p(y=1)$ of predicting class $1$. The datasets consist of $150$ training and $50$ test samples, and has been trained for $5000$ steps with stochastic gradient descent of batch-size $5$, an adaptive learning rate and a square-loss cost function with a gentle $l_2$ regularisation applied to all weights. The loss drops predominantly in the first $200$ steps (left).   }
\label{Fig:fockdb}
\end{figure}

Let us assume we could represent any possible quantum circuit in the feature Hilbert space with the circuit $W(\theta)$. Since the data in $\mathcal{F}$ is linearly separable, there is a $W$ for which we obtain $100\%$ accuracy on the training set, as we saw in Figure \ref{Fig:sq_perc}. However, the goal of machine learning is not to perfectly fit data, but to generalise from it. It is therefore not desirable to find the optimal decision boundary for the training data in $\mathcal{F}$, but to find a good candidate from a class of decision boundaries that captures the structure in the data well.  Such a restricted class of decision boundaries can be defined by using an ansatz for the model circuit $W(\theta)$ which cannot represent \textit{any} circuit, yet still flexible enough to reach interesting candidates. Figure \ref{Fig:graphicandcircuit} c.) shows such a model circuit for the $2$ input modes in our continuous-variable example. The architecture consists of repetitions of a general gate block. We denote by $\hat{a}_{1,2}, \hat{a}^{\dagger}_{1,2}$  the creation and annihilation operators of mode $1$ and $2$, and with $\hat{x}_{1,2}, \hat{p}_{1,2}$ the corresponding quadrature operators (see \cite{weedbrook12}). After an entangling beam splitter gate,
\[BS(u, v) = \e^{u (\e^{iv} \hat{a}_1^{\dagger}\hat{a}_2 -\e^{-iv} \hat{a}_1\hat{a}^{\dagger}_2}), \]
with $u, v \in \mathbb{R}$, the circuit consists of single-mode gates that are first, second and third order in the quadratures. The first-order gate is implemented by a displacement gate
\[D(z) = \e^{\sqrt{2}i(\mathrm{Im}(z) \hat{x} - \mathrm{Re}(z)\hat{p})},\]
with the complex displacement factor $z$. We use a quadratic phase gate for the second order,
\[P(u) = \e^{i\frac{u}{2} \hat{x}^2}, \]
and a cubic phase gate for the third order operator,
\[V(u) = \e^{i\frac{u}{3} \hat{x}^3}.\]
We can in principle construct any continuous-variable quantum circuit from this gate set. This basic circuit block can easily be generalised to circuits of more modes by replacing the single beam splitter by a full optical network of beam splitters \cite{flamini17}. \\

To show that the Fock space classifier works, we plot the decision boundary for the `moons' data in Figure \ref{Fig:fockdb}, using $4$ repetitions of the gate block from Figure \ref{Fig:graphicandcircuit} c.) and $32$ parameters in total. The training loss shows that after about $200$ iterations of a stochastic gradient descent algorithm, the loss converges to almost zero.

\section{Conclusion}

In this paper we introduced a number of new ideas for the area of quantum machine learning based on the theory of feature spaces and kernels. Interpreting the encoding of inputs into quantum states as a feature map, we associate a quantum Hilbert space with a feature space. Inner products of quantum states in this feature space can be used to evaluate a kernel function. We can alternatively train a variational quantum circuit as an explicit classifier in feature space to learn a decision boundary. We introduced a squeezing feature map as an example and motivated with small-scale simulations that these two approach can lead to interesting results. \\

From this work there are many further avenues of research. For example, we raised the question whether there are interesting kernel functions that can be computed by estimating the inner products of quantum states, for which state preparation is classically intractable. Another open question are the details in the design and training of variational circuits, and how learning algorithms can be tailormade for the use in hybrid training schemes. This is a topic that has just begun to be investigated by the quantum machine learning community \cite{verdon17, farhi18}. \\  

Last but not least, we want to come back to a point we raised in the introduction. In quantum machine learning, a lot of models use amplitude encoding, which means that a data vector is represented by the amplitudes of a quantum state. Especially when trying to reproduce neural network-like dynamics one would like to perform nonlinear transformations on the data. But while linear transformations are natural for quantum theory, nonlinearities are difficult to design in this context. Interesting workarounds based on postselection or repeat-until-success circuits were proposed in \cite{wiebe15c,guerreschi17}, but at the considerable costs of making the circuit non-deterministic, and with a probability of failure that grows with the size of the architecture. The feature map approach `outsources' the nonlinearity into the procedure of encoding inputs into a quantum state and therefore offers an elegant solution to the problem of nonlinearities in amplitude encoding. 

\bibliography{litArchive}

\appendix

\section{Reproducing kernels of quantum systems}\label{App:quantumrkhs}

In this section of the appendix we will try to find an answer to the question of which reproducing kernels the Hilbert space of generic quantum systems gives rise to. Quantum theory prescribes that the state of a quantum system is modelled by a vector in a Hilbert space $\mathcal{H}_s$. In a typical setting, the Hilbert space is constructed from a complete basis of eigenvectors $\{\ket{s}\}$ of a complete set of commuting Hermitian operators which corresponds to physical observables. Due to the hermiticity of the observables, the basis is orthogonal, and it can be continuous (i.e., if the observable is the position operator describing the location of a particle), countably infinite (i.e., observing the number of photons in an electric field), or finite (i.e., observing the spin of an electron). Vectors in the Hilbert space are abstractly referred to as $\ket{\psi} \in \mathcal{H}$ in Dirac notation. However, every such Hilbert space has a \textit{functional representation}. In the case of a discrete basis of dimension $N \in \mathbb{N} \cup \infty$, the functional representation $\mathcal{H}^f_s$ of $\mathcal{H}_s$ is given by the (Hilbert) space $l^2$ of square summable sequences $\{ \psi(s_i) =\braket{s_i}{\psi} \}_{i=1}^N $ with the inner product $\langle \psi, \varphi \rangle = \sum_{s_i} \psi(s_i)^* \varphi(s_i) $. In the continuous case this is the space $L^2$ of square summable (equivalence classes of) functions $\psi(s) = \braket{s}{\psi}$ with the inner product $\langle \psi, \varphi \rangle = \int ds \psi(s)^*\varphi(s)  $. The preceding formulation of quantum theory therefore associates every quantum system with a Hilbert space of functions mapping from a set $\mathcal{S}=\{s\}$ to the complex numbers. The question is if these Hilbert spaces give rise to a reproducing kernel that makes them a RKHS with respect to the input set $\mathcal{S}$. \\   

With the resolution of identity $\mathbbm{1} = \int ds\; \ketbra{s}{s} $ for the continuous and $\mathbbm{1} = \sum_i \ketbra{s_i}{s_i}$ for the discrete case, we can immediately ``create'' the reproducing property from Eq. (\ref{Eq:repro}). Consider first the discrete case:
\[ \psi(s_i) = \braket{s_i}{\psi} =  \sum_{s_j} \braket{s_i}{s_j} \braket{s_j}{\psi} = \braket{\braket{s}{\cdot}}{\psi(\cdot) }. \]
We can identify $\braket{s_i}{s_j} $ with the reproducing kernel. Since the basis is orthonormal, we have $\kappa(s_i, s_j) = \delta_{i,j}$. The continuous case is more subtle. Inserting the identity, we get 
\[\psi(s) = \int ds' \braket{s}{s'} \braket{s'}{\psi} = \langle \braket{s}{\cdot}, \psi(\cdot) \rangle , \]
which is the reproducing kernel property with the reproducing kernel $\kappa(s,s') = \braket{s}{s'} $. However, the ``function'' $s'(s) =  \delta(s-s')$ is not square integrable, which means it is itself not part of $\mathcal{H}^f_s$, and the properties of Definition \ref{Def:rkhs} are not fulfilled. This is no surprise, as the space of square integrable functions $L^2$ is a frequent example of a Hilbert space that is not a RKHS \cite{berlinet11}. The inconsistency between Dirac's formalism and functional analysis is also a well-known issue in quantum theory, but usually glossed over in physical contexts \cite{griffiths08}. If mathematical rigour is needed, physicists usually refer to the theory of rigged Hilbert spaces \cite{madrid05}.\\

There are quantum systems with an infinite basis which naturally give rise to a reproducing kernel that is not the delta function. These systems are described by so-called \textit{generalised coherent states} \cite{klauder85}. In the context of quantum machine learning, this has been discussed in Ref. \cite{chatterjee16}. Generalised coherent states are vectors $\ket{l}$ in a Hilbert space $\mathcal{H}_c$ of finite or countably infinite dimension, and where the index $l$ is from some topological space $\mathcal{L}$ (allowing us to define a norm $|| \ket{l} || = \sqrt{\braket{l}{l}}$). They have two fundamental properties. First, $\ket{l}$ is a strongly continuous function of $l$,
\[ \lim_{l\rightarrow l'} ||\; \ket{l'} - \ket{l} || = 0, \; \ket{l} \neq 0.\]
Note that this excludes for example the discrete Fock basis $\{\ket{n}\}$, but also any orthonormal set of states $\{\ket{z}\}$ with a continuous label $z\in \mathbb{C}$, since $\frac{1}{2}|| \ket{z'} - \ket{z} || = 1$ for $z' \neq z$. Second, there exists a measure $\mu$ on $\mathcal{L}$ so that we have a resolution of identity
$ \mathbbm{1} = \int_{\mathcal{L}} \ketbra{l}{l} \; d\mu(l)$. 
This leads to a functional representation of the Hilbert space where a vector $\ket{\psi}\in \mathcal{H}_c$ is expressed via $\ket{\psi} = \sum_l \psi(l) \ket{l}$ with $\psi(l) = \braket{l}{\psi}$. Inserting the resolution of identity to the right hand side of this expression yields
\[\psi(l) =  \int_{\mathcal{L}} \braket{l}{l'}  \braket{l'}{\psi}\; d\mu(l'), \]
which is exactly the reproducing property in Definition \ref{Def:rkhs} with the reproducing kernel $\kappa(l, l') = \braket{l}{l'}$. Since there is a finite overlap between any two states from the basis, the kernel is not the Dirac delta function, and we do not run into the same problem as for continuous orthogonal bases. Hence, the Hilbert space of coherent states is an RKHS for the input set $\{l\}$.\\

The most well-known type of coherent state are optical coherent states
\[\ket{\alpha} = e^{-{|\alpha|^2\over2}}\sum_{n=0}^{\infty}{\alpha^n\over\sqrt{n!}}|n\rangle, \] 
which are the eigenstates of the non-Hermitian bosonic creation operator $\hat{a}$,  with the associated kernel
\begin{equation}
\kappa(\alpha, \beta) = \braket{\alpha}{\beta} = \e^{- \left( \frac{ |\alpha|^2}{2} +\frac{|\beta|^2}{2} - \alpha \beta   \right)},\label{Eq:gaussian} \end{equation}
whose square is a \textit{radial basis function} or \textit{Gaussian kernel} as remarked in \cite{chatterjee16}.

\section{Linear separability in Fock space}\label{App:linsep}

If we map the inputs of a dataset $\mathcal{D}$ to a new dataset 
\[\mathcal{D}' = \{ \ket{(c, \x^1)} ,..., \ket{(c, \x^M)} \},\] 
using the squeezing feature map with phase encoding from Eq. (\ref{Eq:squeezingmap}), the feature mapped data vectors in $\mathcal{D}'$ are always linearly separable, which means any assignment of two classes of labels to the data can be separated by a hyperplane in $\mathcal{F}$ (see Figure \ref{Fig:linear}). To show this, first consider the following:
\begin{prop}\label{Prop:linsep_linid} A set of $M$ vectors in $\mathbb{R}^N$ are linearly separable if $M-1$ of them are linear independent. 
\end{prop}
The proof can be found in Appendix \ref{App:proof_prop1}. Proposition \ref{Prop:linsep_linid} tells us that if our data is linearly independent, it is linearly separable. This result is in fact known from statistical learning theory: The VC dimension -- a measure of flexibility or expressive power -- of linear models in $K$ dimensions is $K+1$, which means that a linear model can separate or ``shatter'' $K+1$ data points if we can choose the strategy of how to arrange them, but not the strategy of how they are labelled.\\

If we can show that the squeezing feature map maps vectors to linearly independent states in Fock space, we know that any dataset becomes linearly separable in Fock space. To simplify, lets first see look at the squeezing map of a single mode.
\begin{prop}\label{Prop:squ_li} Given a set of squeezing phases $\{\varphi^1,...,\varphi^M\}$  with $\varphi^m \neq \varphi^{m'}$ for $m = 1,...,M, m \neq m' $ and a hyperparameter $c \in \mathbb{R}$, the squeezed vacuum Fock states $\ket{(c,\varphi^1)},..., \ket{(c, \varphi^M)}$ are linearly independent. \end{prop}
The proof is found in Appendix \ref{App:proof_prop2}. A very similar proof confirms that the proposition also holds true for the sueezing map with absolute value encoding described in Section \ref{linsep}. Symbolic computation of the rank of the design matrix in feature space in Mathematica confirms this result for randomly selected squeezing factors up to $M=10$ and a cutoff dimension that truncates Fock space to $40$ dimensions. \\

For the multimode feature map dealing with input data of dimension higher than one, 
\[\ket{(c, \boldsymbol\varphi^{m})} = \ket{(c,\varphi^{m}_1)} \otimes \hdots \otimes \ket{(c, \varphi^{m}))},\]
and
\[\ket{(c, \boldsymbol\varphi^{m'})} = \ket{(c,\varphi^{m'}_1)} \otimes \hdots \otimes \ket{(c, \varphi^{m'}))}.\]
We have 
\[\braket{(c, \boldsymbol\varphi^{m})}{(c, \boldsymbol\varphi^{m'})} = \prod \limits_{i=1}^N \braket{(c, \varphi_i^{m})}{(c, \varphi_i^{m'})} ,\]
which is $1$ if $\varphi_i^{m} =\varphi_i^{m'} $ for all $i=1,...,N$ and a value other than zero else. The linear independence therefore carries over to multi-dimensional feature maps.

\section{Proof of Proposition \ref{Prop:linsep_linid}}\label{App:proof_prop1}

Let 
\[\mathcal{D} = \{ (x^1,y^1),\cdots, (x^{M},y^M) \}\]
be a dataset of $M$ vectors with $x^m \in \mathbb{R}^N$ for all $m=1,\cdots, M$, and $y \in \{-1,1\}$. The vectors are guaranteed to be linearly separable if  for any assignment of classes $\{-1,1\}$ to labels $y^1,...,y^{M}$ there is a hyperplane defined by parameters $w_1,...,w_N, b$ so that
\begin{equation}
 \mathrm{sgn} (\sum\limits_{i=1}^N  w_i x^m_i +b) = y^m  \;\;  \forall m=1,...,M. 
\label{Eq:cond} \end{equation}
The sign function is a bit tricky, but if we can instead show that the stronger condition
\begin{equation} \sum\limits_{i=1}^N  w_i x^m_i +b = y^m  \;\;  \forall m=1,...,M  
\label{Eq:cond2} \end{equation}
holds for some parameters,  Eq. \ref{Eq:cond} must automatically be satisfied.\\

Equation \ref{Eq:cond2} defines a system of $M$ linear equations with $N+1$ unknowns (namely the variables).
From the theorey of linear algebra we know \cite{hogben06} that there is at least one solution if and only if the rank of the `coefficient matrix' 
\[ [X|1] = \begin{pmatrix}
x_1^1 & \cdots & x_N^1 & 1\\
\vdots & \ddots & \vdots & \vdots\\
x_1^M & \cdots & x_N^M & 1 \\
\end{pmatrix} \]
is equal to the rank of its augmented matrix
\[ [X|1|y] = \begin{pmatrix}
x_1^1 & \cdots & x_N^1 & 1 & y^1\\
\vdots & \ddots & \vdots & \vdots & \vdots\\
x_1^M & \cdots & x_N^M & 1 & y^M\\
\end{pmatrix}. \] 
Remember that the rank of a matrix is the number of linearly independent row (and column) vectors. \\  

If the data vectors are all linearly independent we have that $N\geq M$ (if $N<M$ there would be some vectors that depend on others, because we have more vectors than dimensions), and the rank of $X$ is $\min (M,N)=M$. Augmenting $X$ by stacking any number of column vectors simply increases $N$, which means that it does not change the rank of the matrix. It follows that for $M$ linearly independent data points embedded in a $N$ dimensional space the system has a solution. The data is therefore linearly separable.\\

With this argument we can add more vectors that are linearly dependent until $M=N$. After this, we can in fact add one (but only one) more data point that linearly depends on the others, and still guarantee linear separability. That is because adding one data point makes the row number equal to the column number in $[X|1]$, and adding more columns does not change the rank. In contrast, adding \textit{two} data points means that we have more columns than rows in $[X|1]$, and adding the column for $[X|1|y]$ can indeed change the rank.

\section{Proof of Proposition \ref{Prop:squ_li}} \label{App:proof_prop2}
Let's consider a matrix $M$ where the squeezed states in Fock basis form the rows:
\[ M_{jn}:= \frac{1}{\sqrt{ \mathrm{cosh}(r_j)}} \left( - \e^{i\phi_j} \tanh(r_j)\right)^n \frac{ \sqrt{(2n)!}}{2^n \; n!} \]
We introduce two auxiliary diagonal matrices:
\[D_1:= \mathrm{diag} \{\sqrt{ \cosh(r_j)} \} \]
\[D_2 := \mathrm{diag}\left\{ \frac{n!}{\sqrt{(2n)!}} \right\} \]
Multiplying, we find that the matrix $V:=D_1MD_2$ has matrix elements 
\[V_{jn} = \left(-\frac{1}{2} \e^{i \phi_j} \tanh(r_j)\right)^n. \]
Importantly, $V$ has the structure of a Vandermonde matrix. In particular,  it has determinant 
\[ \det(V) = \frac{1}{2}  \prod_{1\leq i<j\leq n} \left(- \e^{i\phi_i} \tanh(r_i) + \e^{i\phi_j} \tanh(r_j)\right).\]
The only way that $\det(V) = 0$ is if 
\[ \e^{i(\phi_i-\phi_j)} \tanh(r_i) = \tanh(r_j)\]
for some $i=j$. The squeezing feature map with phase encoding prescribes that $r_i = r_j = c$ (and we assume that $c>0$). Thus,  the  only  solution  to  the  above  equation  is  when $\varphi_i=\varphi_j$, which can only be true if the two feature vectors describe the same datapoint, which we excluded in Proposition \ref{Prop:squ_li}. Thus,  $\det(V)>0$,  which  means  that $\det(M)>0$,  and  hence $M$ is  full  rank.   This  means  that  the  columns  of $M$, which are our feature vectors, are linearly independent. Note that the same proof also prescribes that squeezing feature maps with absolute value encoding makes distinct data points linearly independent in Fock space.

%

\end{document}